\documentclass[10pt,journal,epsfig]{IEEEtran}

\usepackage[dvips]{graphicx}
\usepackage{graphicx}
\usepackage{amssymb}
\usepackage{cite}
\usepackage{subfigure}
\usepackage{amsmath}

\usepackage{algorithm}
\usepackage{algorithmic}

\begin{document}

\title{Quantization Design and Channel Estimation for Massive MIMO Systems with One-Bit ADCs}

\author{Feiyu Wang, Jun Fang, Hongbin Li,~\IEEEmembership{Senior
Member,~IEEE}, and Shaoqian Li,~\IEEEmembership{Fellow,~IEEE}
\thanks{Feiyu Wang, Jun Fang, and Shaoqian Li are with the National Key Laboratory
of Science and Technology on Communications, University of
Electronic Science and Technology of China, Chengdu 611731, China,
Email: JunFang@uestc.edu.cn}
\thanks{Hongbin Li is
with the Department of Electrical and Computer Engineering,
Stevens Institute of Technology, Hoboken, NJ 07030, USA, E-mail:
Hongbin.Li@stevens.edu}
\thanks{This work was supported in part by the National Science
Foundation of China under Grant 61522104, and the National Science
Foundation under Grant ECCS-1408182 and Grant ECCS-1609393.}}

\maketitle

\begin{abstract}
We consider the problem of channel estimation for uplink multiuser
massive MIMO systems, where, in order to significantly reduce the
hardware cost and power consumption, one-bit analog-to-digital
converters (ADCs) are used at the base station (BS) to quantize
the received signal. Channel estimation for one-bit massive MIMO
systems is challenging due to the severe distortion caused by the
coarse quantization. It was shown in previous studies that an
extremely long training sequence is required to attain an
acceptable performance. In this paper, we study the problem of
optimal one-bit quantization design for channel estimation in
one-bit massive MIMO systems. Our analysis reveals that, if the
quantization thresholds are optimally devised, using one-bit ADCs
can achieve an estimation error close to (with an increase by a
factor of $\pi/2$) that of an ideal estimator which has access to
the unquantized data. The optimal quantization thresholds,
however, are dependent on the unknown channel parameters. To cope
with this difficulty, we propose an adaptive quantization (AQ)
approach in which the thresholds are adaptively adjusted in a way
such that the thresholds converge to the optimal thresholds, and a
random quantization (RQ) scheme which randomly generate a set of
nonidentical thresholds based on some statistical prior knowledge
of the channel. Simulation results show that, our proposed AQ and
RQ schemes, owing to their wisely devised thresholds, present a
significant performance improvement over the conventional fixed
quantization scheme that uses a fixed (typically zero) threshold,
and meanwhile achieve a substantial training overhead reduction
for channel estimation. In particular, even with a moderate number
of pilot symbols (about 5 times the number of users), the AQ
scheme can provide an achievable rate close to that of the perfect
channel state information (CSI) case.
\end{abstract}

\begin{keywords}
Massive MIMO systems, channel estimation, one-bit quantization
design, Cram\'{e}r-Rao bound (CRB), maximum likelihood (ML)
estimator.
\end{keywords}

%Massive MIMO systems can also improve the energy efficiency and
%enable the use of inexpensive, low-power components
%\cite{NgoLarsson13}.

\section{Introduction}
Massive multiple-input multiple-output (MIMO), also known as
large-scale or very-large MIMO, is a promising technology to meet
the ever growing demands for higher throughput and better
quality-of-service of next-generation wireless communication
systems \cite{RusekPersson13,LarssonEdfors14,ChenSun16}. Massive
MIMO systems are those that are equipped with a large number of
antennas at the base station (BS) simultaneously serving a much
smaller number of single-antenna users sharing the same
time-frequency slot. By exploiting the asymptotic orthogonality
among channel vectors associated with different users, massive
MIMO systems can achieve almost perfect inter-user interference
cancelation with a simple linear precoder and receive combiner
\cite{Marzetta10}, and thus have the potential to enhance the
spectrum efficiency by orders of magnitude.

%Massive MIMO systems can also improve the energy efficiency and
%enable the use of inexpensive, low-power components
%\cite{NgoLarsson13}.

%the implementation of encounters several practical constraints.

Despite all these benefits, massive MIMO systems pose new
challenges for system design and hardware implementation. Due to
the large number of antennas at the BS, the hardware cost and
power consumption could become prohibitively high if we still
employ expensive and power-hungry high-resolution
analog-to-digital convertors (ADCs) \cite{ChenZhao14}. To address
this obstacle, recent studies (e.g.
\cite{RisiPersson14,FanJin15,ZhangDai16,JacobssonDurisi15,WangLi15,WenWang16,ChoiMo16})
considered the use of low-resolution ADCs (e.g. 1-3 bits) for
massive MIMO systems. It is known that the hardware complexity and
power consumption grow exponentially with the resolution (i.e. the
number of bits per sample) of the ADC. Therefore lowering the
resolution of the ADC can effectively reduce the hardware cost and
power consumption. In particular, for the extreme one-bit case,
the ADC becomes a simple analog comparator. Also, automatic gain
control (AGC) is no longer needed when one-bit ADCs are used,
which further simplifies the hardware complexity.

%which requires minimum cost and power consumption

%The performance of massive MIMO with low resolution and especially
%one-bit ADCs has been investigated from various perspectives.

Massive MIMO with low-resolution ADCs has attracted much attention
over the past few years. Great efforts have been made to
understand the effects of low-resolution ADCs on the performance
of MIMO and massive MIMO systems. Specifically, by assuming full
knowledge of channel state information (CSI), the capacity at both
finite and infinite signal-to-noise ratio (SNR) was derived in
\cite{MoHeath15} for one-bit MIMO systems. For massive MIMO
systems with low-resolution ADCs, the spectral efficiency and the
uplink achievable rate were investigated in
\cite{RisiPersson14,FanJin15,ZhangDai16,LiangZhang16} under
different assumptions. The theoretical analyses suggest that the
use of the low cost and low-resolution ADCs can still provide
satisfactory achievable rates and spectral efficiency.

%low-resolution ADCs lead to a decrease in the achievable rate but
%the performance loss can be compensated by

%In \cite{FanJin15}, an approximate analytical expression for the
%uplink achievable rate of a massive MIMO system with low
%resolution ADCs is derived based on the common maximal-ratio
%combining technique at the receivers. The uplink spectral
%efficiency (SE) of massive MIMO systems with low-resolution ADCs
%over Rician fading channels is investigated in \cite{ZhangDai16}.
%However, in most of the analysis in those related works, channel
%state information is assumed to be perfectly known at the base
%station (BS).

%In most them, e.g. assume perfect CSI at the BS in analyzing the
%achievable performance

%between opposite links (downlink and uplink)

%are concerned about the feasibility and performance of
%low-resolution massive MIMO systems from a signal processing
%perspective. In particular, we

%The aforementioned work studied capacity and achievable rates for
%massive MIMO systems using low-resolution ADCs.

In this paper, we consider the problem of channel estimation for
uplink multiuser massive MIMO systems, where one-bit ADCs are used
at the BS in order to reduce the cost and power consumption.
Channel estimation is crucial to support multi-user MIMO operation
in massive MIMO systems
\cite{AdhikaryNam13,ChoiLove14,SunGao15,GaoDai15,FangLi17}. To
reach the full potential of massive MIMO, accurate downlink CSI is
required at the BS for precoding and other operations. Most
literature on massive MIMO systems, e.g.
\cite{Marzetta10,RusekPersson13,YinGesbert13,MullerCottatellucci14},
assumes a time division duplex (TDD) mode in which the downlink
CSI can be immediately obtained from the uplink CSI by exploiting
channel reciprocity. Nevertheless, channel estimation for massive
MIMO systems with one-bit ADCs is challenging since the magnitude
and phase information about the received signal are lost or
severely distorted due to the coarse quantization. It was shown in
\cite{RisiPersson14} that one-bit massive MIMO systems require an
excessively long training sequence (e.g. approximately 50 times
the number of users) to achieve an acceptable performance. The
work \cite{JacobssonDurisi15} showed that for one-bit massive MIMO
systems, a least-squares channel estimation scheme and a
maximum-ratio combining scheme are sufficient to support both
multiuser operation and the use of high-order constellations.
Nevertheless, a long training sequence is still a requirement. To
alleviate this issue, a Bayes-optimal joint channel and data
estimation scheme was proposed in \cite{WenWang16}, in which the
estimated payload data are utilized to aid channel estimation. In
\cite{ChoiMo16}, a maximum likelihood channel estimator, along
with a near maximum likelihood detector, were proposed for uplink
massive MIMO systems with one-bit ADCs.

%with the objective of minimizing the Cram\'{e}r-Rao bound (CRB)

Despite these efforts, channel estimation using one-bit quantized
data still incur much larger estimation errors as compared with
using the original unquantized data, and require considerably
higher training overhead to attain an acceptable estimation
accuracy. To address this issue, in this paper, we study one-bit
quantizer design and examine the impact of the choice of
quantization thresholds on the estimation performance.
Specifically, the optimal design of quantization thresholds as
well as the training sequences is investigated. Note that one-bit
quantization design is an interesting and important issue but
largely neglected by existing massive MIMO channel estimation
studies. In fact, most channel estimation schemes, e.g.
\cite{RisiPersson14,JacobssonDurisi15,WenWang16,ChoiMo16}, assume
a fixed, typically zero, quantization threshold. The optimal
choice of the quantization threshold was considered in
\cite{KochLapidoth13,Verdu02}, but addressed from an
information-theoretic perspective. Our theoretical results reveal
that, given that the quantization thresholds are optimally
devised, using one-bit ADCs can achieve an estimation error close
to (with an increase only by a factor of $\pi/2$) the minimum
achievable estimation error attained by using infinite-precision
ADCs. The optimal quantization thresholds, however, are dependent
on the unknown channel parameters. To cope with this difficulty,
we propose an adaptive quantization (AQ) scheme by which the
thresholds are dynamically adjusted in a way such that the
thresholds converge to the optimal thresholds, and a random
quantization (RQ) scheme which randomly generates a set of
non-identical thresholds based on some statistical prior knowledge
of the channel. Simulation results show that our proposed schemes,
because of their wisely devised quantization thresholds, present a
significant performance improvement over the fixed quantization
scheme that use a fixed (say, zero) quantization threshold. In
particular, the AQ scheme, even with a moderate number of pilot
symbols (about 5 times the number of users), can provide an
achievable rate close to that of the perfect CSI case.

%, through wisely devising quantization thresholds, present a
%significant performance improvement over the fixed quantization
%scheme that use a fixed (say, zero) quantization threshold, and
%thus can help achieve a substantial training overhead reduction.

The rest of the paper is organized as follows. The system model
and the problem of channel estimation using one-bit ADCs are
discussed in Section \ref{sec:system-model}. In Section
\ref{sec:MLE-CRB}, we develop a maximum likelihood estimator and
carry out a Cram\'{e}r-Rao bound analysis of the one-bit channel
estimation problem. The optimal design of quantization thresholds
and the pilot sequences is studied in Section
\ref{sec:optimal-design}. In Section \ref{sec:AQ-RQ}, we develop
an adaptive quantization scheme and a random quantization scheme
for practical threshold design. Simulation results are provided in
Section \ref{sec:experiments}, followed by concluding remarks in
Section \ref{sec:conclusion}.

%along with an analysis of the best achievable estimation
%performance for one-bit massive MIMO systems

%Also, it remains unclear the best achievable channel estimation
%accuracy one can attain by using one-bit quantized data and its
%performance loss as compared with the best achievable estimation
%accuracy using original unquantized measurements.

%with one-bit ADCs at the base station

\section{System Model and Problem Formulation} \label{sec:system-model}
Consider a single-cell uplink multiuser massive MIMO system, where
the BS equipped with $M$ antennas serves $K$ ($M\gg K$)
single-antenna users simultaneously. The channel is assumed to be
flat block fading, i.e. the channel remains constant over a
certain amount of coherence time. The received signal at the BS
can be expressed as
\begin{align}
\boldsymbol{Y}=\boldsymbol{H}\boldsymbol{X}+\boldsymbol{W}
\label{data-model}
\end{align}
where $\boldsymbol{X}\in\mathbb{C}^{K\times L}$ is a training
matrix and its row corresponds to each user's training sequence
with $L$ pilot symbols, $\boldsymbol{H}\in\mathbb{C}^{M\times K}$
denotes the channel matrix to be estimated, and
$\boldsymbol{W}\in\mathbb{C}^{M\times L}$ represents the additive
white Gaussian noise with its entries following a circularly
symmetric complex Gaussian distribution with zero mean and
variance $2\sigma^2$.

To reduce the hardware cost and power consumption, we consider a
massive MIMO system which uses one-bit ADCs at the BS to quantize
the received signal. Specifically, at each antenna, the real and
imaginary components of the received signal are quantized
separately using a pair of one-bit ADCs. Thus in total $2M$
one-bit ADCs are needed. The quantized output of the received
signal, $\boldsymbol{B}\triangleq [b_{m,l}]$, can be written as
\begin{align}
\boldsymbol{B}=\mathcal{Q}(\boldsymbol{Y})
\label{conventional-quantizer}
\end{align}
where $\mathcal{Q}(\boldsymbol{Y})$ is an element-wise operation
performed on $\boldsymbol{Y}$, and for each element of
$\boldsymbol{Y}$, $y_{m,l}$, we have
\begin{align}
\mathcal{Q}(y_{m,l})=\text{sgn}(\Re(y_{m,l}))+j\text{sgn}(\Im(y_{m,l}))
\end{align}
in which $\Re(y)$ and $\Im(y)$ denote the real and imaginary
components of $y$, respectively, and the sign function
$\text{sgn}(\cdot)$ is defined as
\begin{align}
\text{sgn}(y) \triangleq \left\{ \begin{array}{ll}
1 & \textrm{if $y\ge 0$}\\
-1 & \textrm{otherwise}
\end{array} \right.
\end{align}
Therefore the quantized output belongs to the set
\begin{align}
b_{m,l}\in \{1+j,-1+j,1-j,-1-j\}\quad \forall m,l
\end{align}
Note that in (\ref{conventional-quantizer}), we implicitly assume
a zero threshold for one-bit quantization. Nevertheless, using
identically a zero threshold for all measurements is not
necessarily optimal, and it is interesting to analyze the impact
of the quantization thresholds on the channel estimation
performance. Such an issue (i.e. choice of quantization
thresholds), albeit important, was to some extent neglected by
most existing studies. To examine this problem, let
$\boldsymbol{T}\triangleq [\tau_{m,l}]$ denote the thresholds used
for one-bit quantization. The quantized output of the received
signal, $\boldsymbol{B}$, is now given as
\begin{align}
\boldsymbol{B}=\mathcal{Q}(\boldsymbol{Y}-\boldsymbol{T})
\label{quantizer}
\end{align}

To facilitate our analysis, we first convert (\ref{data-model})
into a real-valued form as follows
\begin{align}
\boldsymbol{\widetilde{Y}}=\boldsymbol{\widetilde{A}}\boldsymbol{\widetilde{H}}+\boldsymbol{\widetilde{W}}
\end{align}
where
\begin{align}
\boldsymbol{\widetilde{Y}}\triangleq & [ \Re(\boldsymbol{Y}) \
\Im(\boldsymbol{Y})]^T \nonumber\\
\boldsymbol{\widetilde{H}} \triangleq & [ \Re(\boldsymbol{H}) \
\Im(\boldsymbol{H})]^T \nonumber\\
\boldsymbol{\widetilde{W}} \triangleq & [ \Re(\boldsymbol{W}) \
\Im(\boldsymbol{W})]^T \nonumber
\end{align}
and
\begin{align}
\boldsymbol{\widetilde{A}} \triangleq \left[\begin{array}{ccc}
\Re(\boldsymbol{X}) & \Im(\boldsymbol{X}) \\
-\Im(\boldsymbol{X}) & \Re(\boldsymbol{X})
\end{array}\right]^T \label{A-X-relationship}
\end{align}
Vectorizing the real-valued matrix $\boldsymbol{\widetilde{Y}}$,
the received signal can be expressed as a real-valued vector form
as
\begin{align}
\boldsymbol{y}=\boldsymbol{A}\boldsymbol{h}+\boldsymbol{w}
\label{data-model-vector}
\end{align}
where
$\boldsymbol{y}\triangleq\text{vec}(\boldsymbol{\widetilde{Y}})$,
$\boldsymbol{A}\triangleq\boldsymbol{I}_{M}\otimes\boldsymbol{\widetilde{A}}$,
$\boldsymbol{h}\triangleq\text{vec}(\boldsymbol{\widetilde{H}})$,
and
$\boldsymbol{w}\triangleq\text{vec}(\boldsymbol{\widetilde{W}})$.
It can be easily verified $\boldsymbol{y}\in\mathbb{R}^{2ML}$,
$\boldsymbol{A}\in\mathbb{R}^{2ML\times 2MK}$, and
$\boldsymbol{h}\in\mathbb{R}^{2MK}$. Accordingly, the one-bit
quantized data can be written as
\begin{align}
\boldsymbol{b}=\text{sgn}(\boldsymbol{y}-\boldsymbol{\tau})
\label{quantized-data-model-vector}
\end{align}
where $\boldsymbol{\tau}\triangleq \text{vec}([
\Re(\widetilde{\boldsymbol{T}}) \
\Im(\widetilde{\boldsymbol{T}})]^T)$ and
$\boldsymbol{\tau}\in\mathbb{R}^{2ML}$. For simplicity, we define
$N\triangleq 2ML$. 

Our objective in this paper is to estimate the channel
$\boldsymbol{h}$ based on the one-bit quantized data
$\boldsymbol{b}$, examine the best achievable estimation
performance and investigate the optimal thresholds
$\boldsymbol{\tau}$ as well as the optimal training sequences
$\boldsymbol{X}$. To this objective, in the following, we first
develop a maximum likelihood (ML) estimator and carry out a
Cram\'{e}r-Rao bound (CRB) analysis. The optimal choice of the
quantization thresholds as well as the training sequences is then
studied based on the CRB matrix of the unknown channel parameter
vector $\boldsymbol{h}$.

\section{ML Estimator and CRB Analysis} \label{sec:MLE-CRB}
\subsection{ML Estimator}
By combining (\ref{data-model-vector}) and
(\ref{quantized-data-model-vector}), we have
\begin{align}
b_n=\text{sgn}(y_n-\tau_n)
=\text{sgn}(\boldsymbol{a}_n^{T}\boldsymbol{h}+w_n-\tau_n), \quad
\forall n
\end{align}
where, by allowing a slight abuse of notation, we let $b_n$,
$y_n$, $\tau_n$, and $w_n$ denote the $n$th entry of
$\boldsymbol{b}$, $\boldsymbol{y}$, $\boldsymbol{\tau}$, and
$\boldsymbol{w}$, respectively; and $\boldsymbol{a}_n^T$ denotes
the $n$th row of $\boldsymbol{A}$. It is easy to derive that
\begin{align}
P(b_n=1;\boldsymbol{h}) & =P(w_n\geq-(\boldsymbol{a}_n^{T}\boldsymbol{h}-\tau_n);\boldsymbol{h}) \nonumber \\
& =F_{w}(\boldsymbol{a}_n^{T}\boldsymbol{h}-\tau_n)
\end{align}
and
\begin{align}
P(b_n=-1;\boldsymbol{h}) & =P(w_n < -(\boldsymbol{a}_n^{T}\boldsymbol{h}-\tau_n);\boldsymbol{h}) \nonumber \\
& =1-F_{w}(\boldsymbol{a}_n^{T}\boldsymbol{h}-\tau_n)
\end{align}
where $F_{w}(\cdot)$ denotes the cumulative density function (CDF)
of $w_n$, and $w_n$ is a real-valued Gaussian random variable with
zero-mean and variance $\sigma^2$. Therefore the probability mass
function (PMF) of $b_n$ is given by
\begin{align}
p(b_n;\boldsymbol{h})=& [1-F_{w}(\boldsymbol{a}_n^{T}\boldsymbol{h}-\tau_n)]^{(1-b_n)/2} \nonumber \\
&\cdot[F_{w}(\boldsymbol{a}_n^{T}\boldsymbol{h}-\tau_n)]^{(1+b_n)/2}
\end{align}
Since $\{b_n\}$ are independent, the log-PMF or log-likelihood
function can be written as
\begin{align}
L(\boldsymbol{h}) & \triangleq  \log p(b_1,\dots,b_N;\boldsymbol{h}) \nonumber \\
& = \sum_{n=1}^{N} \bigg\{ \frac{1-b_n}{2}\log [1-F_{w}(\boldsymbol{a}_n^{T}\boldsymbol{h}-\tau_n)]   \nonumber \\
& \qquad \quad + \frac{1+b_n}{2} \log
[F_{w}(\boldsymbol{a}_n^{T}\boldsymbol{h}-\tau_n)] \bigg\}
\label{log-PMF}
\end{align}
The ML estimate of $\boldsymbol{h}$, therefore, is given as
\begin{align}
\hat{\boldsymbol{h}}=\arg\max_{\boldsymbol{h}} \ L(\boldsymbol{h})
\label{MLE}
\end{align}
It can be proved that the log-likelihood function
$L(\boldsymbol{h})$ is a concave function. Hence computationally
efficient search algorithms can be used to find the global
maximum. The proof of the concavity of $L(\boldsymbol{h})$ is
given in Appendix \ref{appA}.

%Specifically, first-order methods with very low computational
%complexity, such as Nesterov's optimal first-order method, can be
%employed to accelerate the convergence rate of the search
%algorithm.

%Note that under mild regularity conditions, the ML estimator is
%asymptotically (in terms of the sample size) unbiased and
%asymptotically achieves the CRB.

%Generally, because of non-concavity of the log-likelihood
%function, ML estimation often suffers from local maxima and has
%high computational complexity. Nonetheless, it is not the case in
%this problem.

\subsection{CRB}
We now carry out a CRB analysis of the one-bit channel estimation
problem (\ref{quantized-data-model-vector}). The CRB results help
understand the effect of different system parameters, including
quantization thresholds as well as training sequences, on the
estimation performance. We first summarize our derived CRB results
in the following theorem.

\newtheorem{theorem}{Theorem}
\begin{theorem} \label{theorem1}
The Fisher information matrix (FIM) for the estimation problem
(\ref{quantized-data-model-vector}) is given as
\begin{align}
\boldsymbol{J}(\boldsymbol{h})=\sum_{n=1}^{N}
g(\tau_n,\boldsymbol{a}_n)\boldsymbol{a}_n\boldsymbol{a}_n^T
\end{align}
where $g(\tau_n,\boldsymbol{a}_n)$ is defined as
\begin{align}
g(\tau_n,\boldsymbol{a}_n) \triangleq \frac {f_{w}^2
(\boldsymbol{a}_n^{T}\boldsymbol{h}-\tau_n)}
{F_{w}(\boldsymbol{a}_n^{T}\boldsymbol{h}-\tau_n)(1-F_{w}(\boldsymbol{a}_n^{T}\boldsymbol{h}-\tau_n))}
\label{g-function}
\end{align}
in which $f_{w}(\cdot)$ denotes the probability density function
(PDF) of $w_n$. Accordingly, the CRB matrix for the estimation
problem (\ref{quantized-data-model-vector}) is given by
\begin{align}
\text{CRB}(\boldsymbol{h})=\boldsymbol{J}^{-1}(\boldsymbol{h}) =
\left( \sum_{n=1}^{N} g(\tau_n,\boldsymbol{a}_n) \boldsymbol{a}_n
\boldsymbol{a}_n^T \right)^{-1} \label{CRB}
\end{align}
\end{theorem}
\begin{proof}
See Appendix \ref{appB}.
\end{proof}

As is well known, the CRB places a lower bound on the estimation
error of any unbiased estimator \cite{Kay93} and is asymptotically
attained by the ML estimator. Specifically, the covariance matrix
of any unbiased estimate satisfies:
$\text{cov}(\hat{\boldsymbol{h}})-\text{CRB}(\boldsymbol{h})
\succeq \boldsymbol{0}$. Also, the variance of each component is
bounded by the corresponding diagonal element of
$\text{CRB}(\boldsymbol{h})$, i.e., $\text{var}(\hat{h}_i) \ge
[\text{CRB}(\boldsymbol{h})]_{ii}$.

We observe from (\ref{CRB}) that the CRB matrix of
$\boldsymbol{h}$ depends on the quantization thresholds
$\boldsymbol{\tau}$ as well as the matrix $\boldsymbol{A}$ which
is constructed from training sequences $\boldsymbol{X}$ (cf.
(\ref{A-X-relationship})). Naturally, we wish to optimize
$\boldsymbol{\tau}$ and $\boldsymbol{A}$ (i.e. $\boldsymbol{X}$)
by minimizing the trace of the CRB matrix, i.e. the overall
estimation error asymptotically achieved by the ML estimator. The
optimization therefore can be formulated as follows
\begin{align}
\min_{\boldsymbol{X},\boldsymbol{\tau}}\quad &
\text{tr}\left\{\text{CRB}(\boldsymbol{h})\right\} = \mathrm{tr}
\left\{ \left( \sum_{n=1}^{N} g(\tau_n,\boldsymbol{a}_n)
\boldsymbol{a}_n \boldsymbol{a}_n^T \right)^{-1} \right\}
\nonumber\\
\text{s.t.} \quad &
\boldsymbol{A}=\boldsymbol{I}_M\otimes\boldsymbol{\widetilde{A}}
\nonumber\\
& \boldsymbol{\widetilde{A}} \triangleq \left[\begin{array}{ccc}
\Re(\boldsymbol{X}) & \Im(\boldsymbol{X}) \\
-\Im(\boldsymbol{X}) & \Re(\boldsymbol{X})
\end{array}\right]^T \nonumber\\
& \text{tr}(\boldsymbol{X}\boldsymbol{X}^H)\leq P
 \label{opt1}
\end{align}
where $\text{tr}(\boldsymbol{X}\boldsymbol{X}^H)\leq P$ is a
transmit power constraint imposed on the pilot signals. Such an
optimization is examined in the following section, where it is
shown that the optimization of $\boldsymbol{X}$ can be decoupled
from the optimization of the threshold $\boldsymbol{\tau}$.

\begin{figure}[!t]
\centering
\includegraphics[width=3.5in]{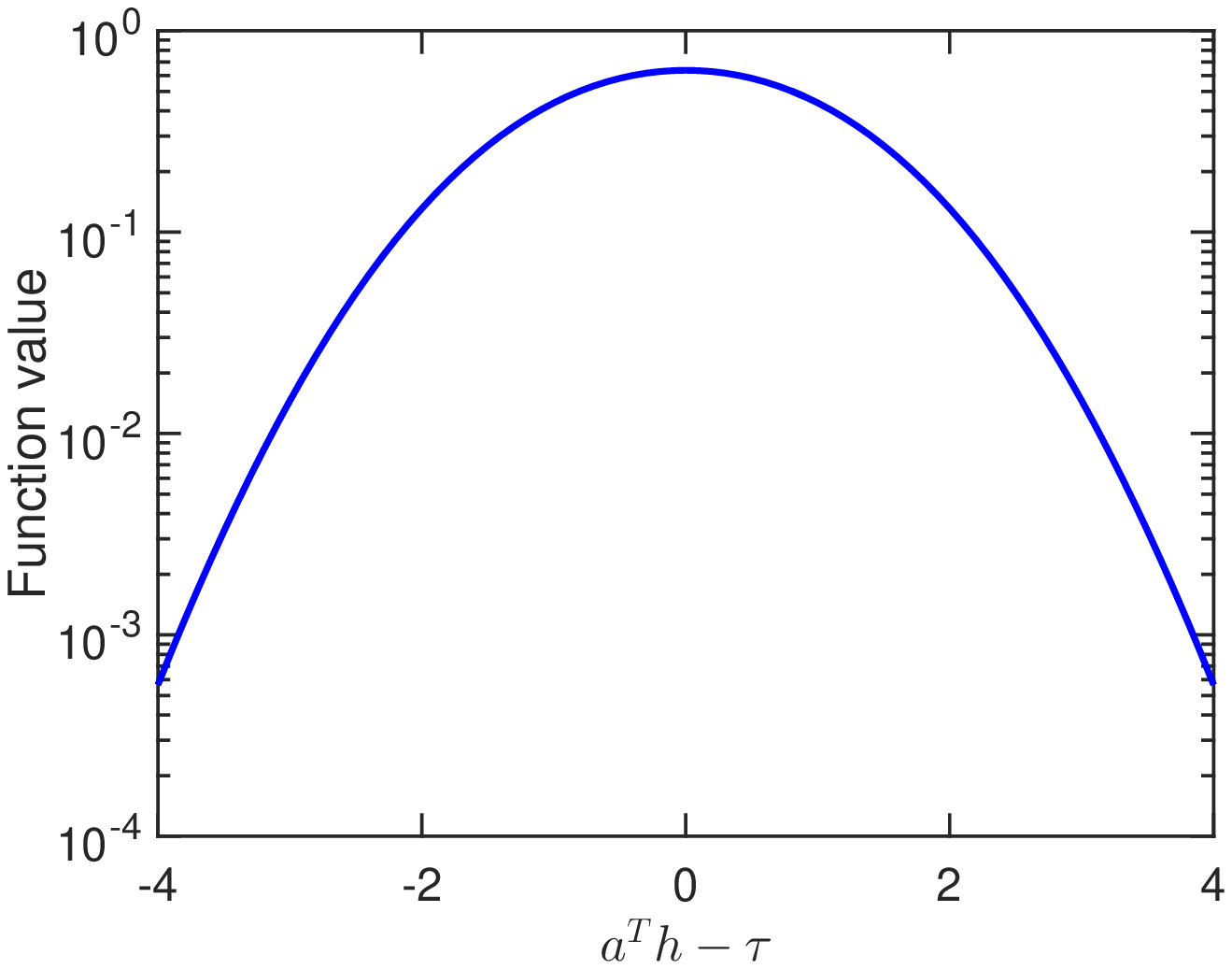}
\caption{The function value of $g(\tau_n,\boldsymbol{a}_n)$ vs.
$(\boldsymbol{a}_n^T\boldsymbol{h}-\tau_n)$, where $\sigma^2=1$.}
\label{fig1}
\end{figure}

%has been shown in previous distributed estimation studies (e.g.)

\section{Optimal Design and Performance Analysis} \label{sec:optimal-design}
\subsection{Optimal Quantization Thresholds and Pilot
Sequences} Before proceeding, we first introduce the following
result.

\newtheorem{proposition}{Proposition}
\begin{proposition} \label{proposition1}
For the Gaussian random variable $w_n$,
$g(\tau_n,\boldsymbol{a}_n)$ defined in (\ref{g-function}) is a
positive and symmetric function attaining its maximum when
$\tau_n=\boldsymbol{a}_n^{T}\boldsymbol{h}$ (see Fig. \ref{fig1}).
\end{proposition}
\begin{proof}
Please see Appendix \ref{appE}.
\end{proof}

Hence, given a fixed $\boldsymbol{A}$ (i.e. $\boldsymbol{X}$), the
optimal quantization thresholds conditional on $\boldsymbol{A}$
are given by
\begin{align}
\tau_n^{\star}=\boldsymbol{a}_n^{T}\boldsymbol{h}, \quad \forall
n\in\{1,\ldots,N\} \label{optimumthreshold}
\end{align}
The result (\ref{optimumthreshold}) comes directly by noting that
\begin{align}
\sum_{n=1}^N
g_n(\tau_n^{\star},\boldsymbol{a}_n)\boldsymbol{a}_n\boldsymbol{a}_n^T-\sum_{n=1}^N
g_n(\tau_n,\boldsymbol{a}_n)\boldsymbol{a}_n\boldsymbol{a}_n^T\succeq\mathbf{0}
\end{align}
and resorting to the convexity of $\text{tr}(\boldsymbol{P}^{-1})$
over the set of positive definite matrix, i.e. for any
$\boldsymbol{P}\succ\mathbf{0}$, $\boldsymbol{Q}\succ\mathbf{0}$,
and $\boldsymbol{P}-\boldsymbol{Q}\succeq \mathbf{0}$, the
following inequality
$\text{tr}(\boldsymbol{P}^{-1})\leq\text{tr}(\boldsymbol{Q}^{-1})$
holds (see \cite{BoydVandenberghe03}).

We see that the optimal choice of the quantization threshold
$\tau_n$ is dependent on the unknown channel $\boldsymbol{h}$. To
facilitate our analysis, we, for the time being, suppose
$\boldsymbol{h}$ is known. Substituting (\ref{optimumthreshold})
into (\ref{opt1}) and noting that
\begin{align}
g(\tau_n^{\star},\boldsymbol{a}_n)=\frac {f_{w}^2 (0)}
{F_{w}(0)(1-F_{w}(0))}  = \frac{2}{\pi\sigma^2} \quad \forall n
\end{align}
the optimization (\ref{opt1}) reduces to
\begin{align}
\min_{\boldsymbol{X}} \quad & \frac{\pi\sigma^2}{2} \text{tr}
\left\{ \left( \boldsymbol{A}^T \boldsymbol{A} \right)^{-1}
\right\}
\nonumber\\
\text{s.t.} \quad &
\boldsymbol{A}=\boldsymbol{I}_M\otimes\boldsymbol{\widetilde{A}}
\nonumber\\
& \boldsymbol{\widetilde{A}} \triangleq \left[\begin{array}{ccc}
\Re(\boldsymbol{X}) & \Im(\boldsymbol{X}) \\
-\Im(\boldsymbol{X}) & \Re(\boldsymbol{X})
\end{array}\right]^T \nonumber\\
& \text{tr}(\boldsymbol{X}\boldsymbol{X}^H)\leq P \label{opt2}
\end{align}
which is now independent of $\boldsymbol{h}$. We have the
following theorem regarding the solution to the optimization
(\ref{opt2}).

%Note that from the constraint
%$\text{tr}(\boldsymbol{X}\boldsymbol{X}^H)\leq P$, we can easily
%derive that
%\begin{align}
%\text{tr}(\boldsymbol{A}\boldsymbol{A}^H)\leq P
%\end{align}

\begin{theorem} \label{theorem2}
The minimum achievable objective function value of (\ref{opt2}) is
given by $(\pi\sigma^2 MK^2)/P$ and can be attained if the pilot
matrix $\boldsymbol{X}$ satisfies
\begin{align}
\boldsymbol{X}\boldsymbol{X}^H = (P/K) \boldsymbol{I}
\label{theorem2:eqn1}
\end{align}
\end{theorem}
\begin{proof}
See Appendix \ref{appC}.
\end{proof}

Theorem \ref{theorem2} reveals that, for one-bit massive MIMO
systems, users should employ orthogonal pilot sequences in order
to minimize channel estimation errors. Although it is a convention
to use orthogonal pilots to facilitate channel estimation for
conventional massive MIMO systems, to our best knowledge, its
optimality in one-bit massive MIMO systems has not been
established before.

\subsection{Performance Analysis}
We now investigate the estimation performance when the optimal
thresholds are employed, and its comparison with the performance
attained by a conventional massive MIMO system which assumes
infinite-precision ADCs. Substituting the optimal thresholds
(\ref{optimumthreshold}) into the CRB matrix (\ref{CRB}), we have
\begin{align}
\text{CRB}_{\text{OQ}}(\boldsymbol{h})=\frac{\pi\sigma^2}{2}
\left( \boldsymbol{A}^T \boldsymbol{A} \right)^{-1} \label{CRB-Q}
\end{align}
where for clarity, we use the subscript OQ to represent the
estimation scheme using optimal quantization thresholds. On the
other hand, when the unquantized observations $\boldsymbol{y}$ are
available, it can be readily verified that the CRB matrix is given
as
\begin{align}
\text{CRB}_{\text{NQ}}(\boldsymbol{h})=\sigma^2 \left(
\boldsymbol{A}^T \boldsymbol{A} \right)^{-1} \label{CRB-NQ}
\end{align}
where we use the subscript NQ to represent the scheme which has
access to the unquantized observations. Comparing (\ref{CRB-Q})
with (\ref{CRB-NQ}), we can see that if optimal thresholds are
employed, then using one-bit ADCs for channel estimation incurs
only a mild performance loss relative to using infinite-precision
ADCs, with the CRB increasing by only a factor of $\pi/2$, i.e.
\begin{align}
\text{CRB}_{\text{OQ}}(\boldsymbol{h})=\frac{\pi}{2}
\text{CRB}_{\text{NQ}}(\boldsymbol{h})
\end{align}
We also take a glimpse of the estimation performance as the
thresholds deviate from their optimal values. For simplicity, let
$\tau_n=\tau_n^{\star}+\delta=\boldsymbol{a}_n^{T}\boldsymbol{h}+\delta,
\forall n$, in which case the CRB matrix is given by
\begin{align}
\text{CRB}_{\text{Q}}(\boldsymbol{h})=\frac
{F_{w}(\delta)(1-F_{w}(\delta))}{f_{w}^2 (\delta)}\left(
\boldsymbol{A}^T \boldsymbol{A} \right)^{-1}
\end{align}
Since $(F_{w}(\delta)(1-F_{w}(\delta)))/f_{w}^2(\delta)$ is the
reciprocal of $g(\tau_n,\boldsymbol{a}_n)$, from Fig. \ref{fig1},
we know that the function value
$(F_{w}(\delta)(1-F_{w}(\delta)))/f_{w}^2(\delta)$ grows
exponentially as $|\delta|$ increases. This indicates that a
deviation of the thresholds from their optimal values results in a
substantial performance loss.

%can be tightly bounded by
%\begin{align}
%\frac {F_{w}(\delta)(1-F_{w}(\delta))}{f_{w}^2
%(\delta)}\leq\frac{\pi\sigma^2}{2}\exp\left(\frac{\delta^2}{2\sigma^2}\right)
%\end{align}

In summary, the above results have important implications for the
design of one-bit massive MIMO systems. It points out that a
careful choice of quantization thresholds can help improve the
estimation performance significantly, and help achieve an
estimation accuracy close to an ideal estimator which has access
to the raw observations $\boldsymbol{y}$.

%which is an important factor but neglected in most previous
%studies

The problem lies in that the optimal thresholds
$\boldsymbol{\tau}$ are functions of $\boldsymbol{h}$, as
described in (\ref{optimumthreshold}). Since $\boldsymbol{h}$ is
unknown and to be estimated, the optimal thresholds
$\boldsymbol{\tau}$ are also unknown. To address this difficulty,
we, in the following, propose an adaptive quantization (AQ) scheme
by which the thresholds are dynamically adjusted from one
iteration to another, and a random quantization (RQ) schme which
randomly generates a set of nonidentical thresholds based on some
statistical prior knowledge of the channel.

\section{Practical Threshold Design Strategies} \label{sec:AQ-RQ}
\subsection{Adaptive Quantization}
One strategy to overcome the above difficulty is to use an
iterative algorithm in which the thresholds are iteratively
refined based on the previous estimate of $\boldsymbol{h}$.
Specifically, at iteration $i$, we use the current quantization
thresholds $\boldsymbol{\tau}^{(i)}$ to generate the one-bit
observation data $\boldsymbol{b}^{(i)}$. Then a new estimate
$\hat{\boldsymbol{h}}^{(i)}$ is obtained from the ML estimator
(\ref{MLE}). This estimate is then plugged in
(\ref{optimumthreshold}) to obtain updated quantization
thresholds, i.e. $\boldsymbol{\tau}^{(i+1)}=\boldsymbol{A}
\hat{\boldsymbol{h}}^{(i)}$, for subsequent iteration. When
computing the ML estimate $\hat{\boldsymbol{h}}^{(i)}$, not only
the quantized data from the current iteration but also from all
previous iterations can be used. The ML estimator (\ref{MLE}) can
be easily adapted to accommodate these quantized data since the
data are independent across different iterations. Due to the
consistency of the ML estimator for large data records, this
iterative process will asymptotically lead to optimal quantization
thresholds, i.e. $\boldsymbol{\tau}^{(i)} \stackrel{i \to
\infty}{\longrightarrow} \boldsymbol{A} \boldsymbol{h}$. In fact,
our simulation results show that the adaptive quantization scheme
yields quantization thresholds close to the optimal values within
only a few iterations.

For clarity, we summarize the adaptive quantization (AQ) scheme as
follows.

\begin{center}
\textbf{Adaptive Quantization Scheme}
\end{center}
\vspace{0cm} \noindent
\begin{tabular}{lp{7.7cm}}
\hline 1.& Select an initial quantization threshold
$\boldsymbol{\tau}^{(0)}$ and the maximum number of iterations $i_{\text{max}}$. \\
2.& At iteration $i=1,2,\ldots$: Based on $\boldsymbol{y}$ and
$\boldsymbol{\tau}^{(i)}$, calculate the new binary data
$\boldsymbol{b}^{(i)}=\text{sgn}(\boldsymbol{y}-\boldsymbol{\tau}^{(i)})$. \\
3.& Compute a new estimate of $\boldsymbol{h}$,
$\hat{\boldsymbol{h}}^{(i)}$,
via (\ref{MLE}). \\
4.& Calculate new thresholds according to $\boldsymbol{\tau}^{(i+1)}=\boldsymbol{A}\hat{\boldsymbol{h}}^{(i)}$. \\
5.& Go to Step 2 if $i < i_{\text{max}}$. \\
\hline
\end{tabular}

%A rigorous proof of this asymptotic optimality was provided in [],
%where a similar adaptive iterative algorithm was proposed to
%adjust the quantization thresholds in WSNs.

Note that during the iterative process, the channel
$\boldsymbol{h}$ is assumed constant over time. Thus the AQ scheme
can be used to estimate channels that are unchanged or slowly
time-varying across a number of consecutive frames. For example,
for the scenario where the relative speeds between the mobile
terminals and the base station are slow, say, 2 meters per second,
the channel coherence time could be up to tens of milliseconds,
more precisely, about 60 milliseconds if the carrier frequency is
set to 1GHz, according to the Clarke's model
\cite{TseViswanath05}. Suppose the time duration of each frame is
10 milliseconds which is a typical value for practical LTE
systems. In this case, the channel remains unchanged across 6
consecutive frames. We can use the AQ scheme to update the
quantization thresholds at each frame based on the channel
estimate obtained from the previous frame. In this way, we can
expect that the quantization thresholds will come closer and
closer to the optimal values from one frame to the next, and as a
result, a more and more accurate channel estimate can be obtained.

\begin{figure}[!t]
\centering
\includegraphics[width=3.5in]{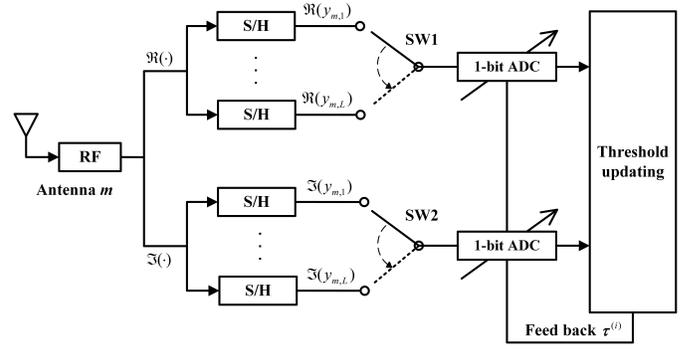}
\caption{An off-line implementation of the AQ scheme.}
\label{fig2}
\end{figure}

The above scheme assumes a static or slowly time-varying channel
across multiple frames. Another way of implementing the AQ scheme
requires no such an assumption, but at the expense of increased
hardware complexity. The idea is to use a number of
sample-and-hold (S/H) circuits to sample the analog received
signals and to store their values for subsequent offline
processing. Specifically, each antenna/RF chain is followed by
$2L$ S/H circuits which are equally divided into two groups to
sample and store the real and imaginary components, respectively
(see Fig. \ref{fig2}). Through a precise timing control, we ensure
that at each antenna, say, the $m$th antenna, the $l$th S/H
circuit pair in the two groups are controlled to store the real
and imaginary components of the $l$th received pilot symbol, i.e.
$\Re(y_{m,l})$ and $\Im(y_{m,l})$, respectively. Also, to avoid
using a one-bit ADC for each S/H circuit, a switch can be used to
connect a single one-bit ADC with multiple S/H circuits. Once the
analog signals $\boldsymbol{y}$ have been stored, the AQ scheme
can be implemented in an offline manner. Clearly, this offline
approach can be implemented on a single frame basis, and thus no
longer requires a static channel assumption. Nevertheless, such an
implementation requires a number of S/H circuits as well as
precise timing control for sampling and quantization. Also, this
offline processing may cause a latency issue which should be taken
care of in practical systems.

%such that the one-bit ADC can quantize the analog signals in a
%sequential manner with possibly different thresholds.

%We now discuss the feasibility of the AQ scheme in practical
%systems.

%of the received signal the real components and the remaining $L$
%S/H circuits are for the imaginary components

%which the first half $L$ S/H circuits are for

%Another strategy, without resorting to the iterative algorithm,

\subsection{Random Quantization}
The AQ scheme requires the channel to be (approximately)
stationary, or needs to be implemented with additional hardware
circuits. Here we propose a random quantization (RQ) scheme that
does not involve any iterative procedure and is simple to
implement. The idea is to randomly generate a set of non-identical
thresholds based on some statistical prior knowledge of
$\boldsymbol{h}$, with the hope that some of the thresholds are
close to the unknown optimal thresholds. For example, suppose each
entry of $\boldsymbol{h}$ follows a Gaussian distribution with
zero mean and variance $\sigma_h^2$. Note that different entries
of $\boldsymbol{h}$ may have different variances due to the reason
that they may correspond to different users. Nevertheless, we
assume the same variance for all entries for simplicity. We
randomly generate $N$ different realizations of $\boldsymbol{h}$,
denoted as $\{\boldsymbol{\tilde{h}}_n\}$, following this known
distribution. The $N$ quantization thresholds are then devised
according to
\begin{align}
\tau_n=\boldsymbol{a}_n^{T}\boldsymbol{\tilde{h}}_n, \quad \forall
n\in\{1,\ldots,N\} \label{multi-thresholding}
\end{align}
Our simulation results suggest that this RQ scheme can achieve a
considerable performance improvement over the conventional fixed
quantization scheme which uses a fixed (typically zero) threshold.
The reason is that the thresholds produced by
(\ref{multi-thresholding}) are more likely to be close to their
optimal values.

%Compared with the fixed thresholds-based channel estimation
%scheme, when the knowledge of statistical properties is available
%for the receiver, the thresholds could be determined according to
%a certain heuristic rule and the channel estimation performance
%can be significantly improved through exploiting the prior
%knowledge.

%Suppose the optimal pilot sequences are adopted. Then the overall
%channel estimation errors asymptotically achieved by a ML
%estimator for one-bit massive MIMO and infinite-precision massive
%MIMO systems are given respectively as follows,

%the ML estimator is given by
%\begin{align}
%\hat{\boldsymbol{h}}_{\text{NQ}} = \left(
%\boldsymbol{A}^T\boldsymbol{A} \right)^{-1} \boldsymbol{A}^T
%\boldsymbol{y}
%\end{align}

\section{Simulation Results} \label{sec:experiments}
We now carry out experiments to corroborate our theoretical
analysis and to illustrate the performance of our proposed one-bit
quantization schemes, i.e. the AQ and the RQ schemes. We compare
our schemes with the conventional fixed quantization (FQ) scheme
which employs a fixed zero threshold for one-bit quantization, and
a no-quantization scheme (referred to as NQ) which uses the
original unquantized data for channel estimation. For the NQ
scheme, it can be easily verified that its ML estimate is given by
\begin{align}
\boldsymbol{\hat{h}}=(\boldsymbol{A}^T\boldsymbol{A})^{-1}\boldsymbol{A}^T\boldsymbol{y}
\end{align}
and its associated CRB is given by (\ref{CRB-NQ}). For other
schemes such as the RQ and the FQ, although a close-form
expression is not available, the ML estimate can be obtained by
solving the convex optimization (\ref{MLE}). In our simulations,
we assume independent and identically distributed (i.i.d.)
rayleigh fading channels, i.e. all elements of $\boldsymbol{H}$
follow a circularly symmetric complex Gaussian distribution with
zero mean and unit variance. Training sequences $\boldsymbol{X}$
which satisfy (\ref{theorem2:eqn1}) are randomly generated. The
signal-to-noise ratio (SNR) is defined as
\begin{align}
\text{SNR}=\frac{P}{KL\sigma^2}
\end{align}

\begin{figure}[!t]
 \centering
\begin{tabular}{c}
\includegraphics[width=3.5in]{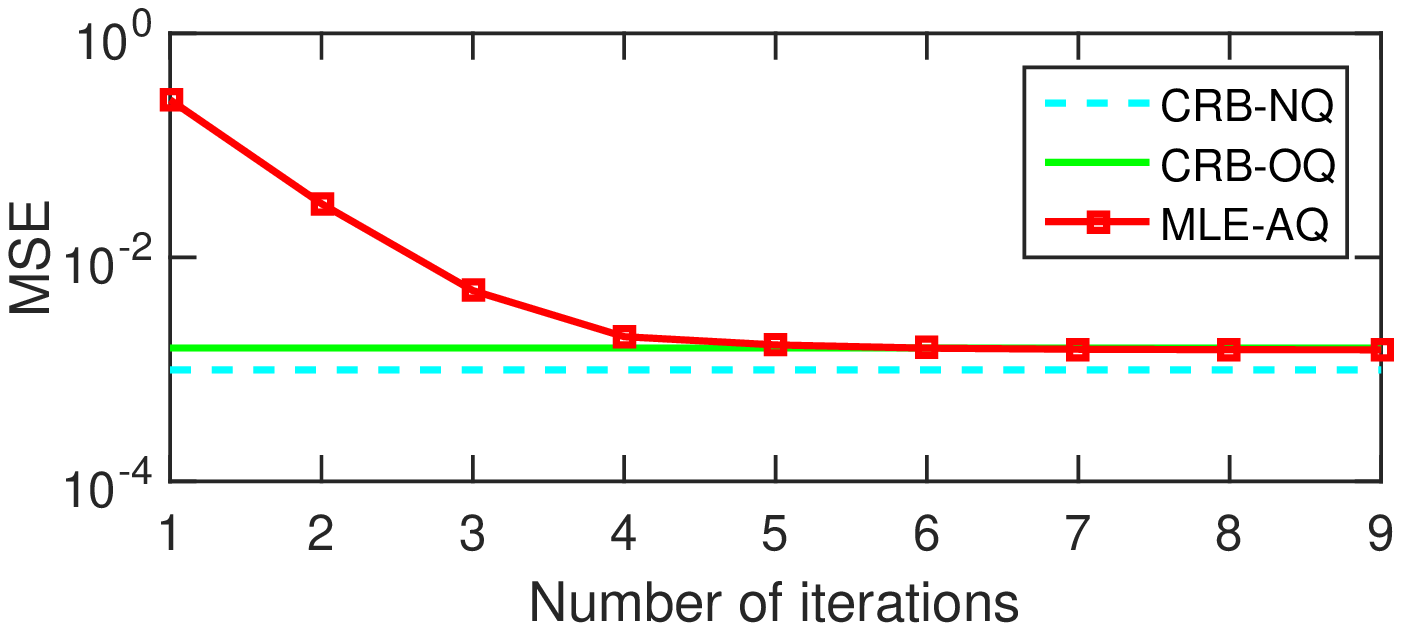}\\
(a). $K=8$, $L=32$ and $\text{SNR}=15$ dB. \\
\includegraphics[width=3.5in]{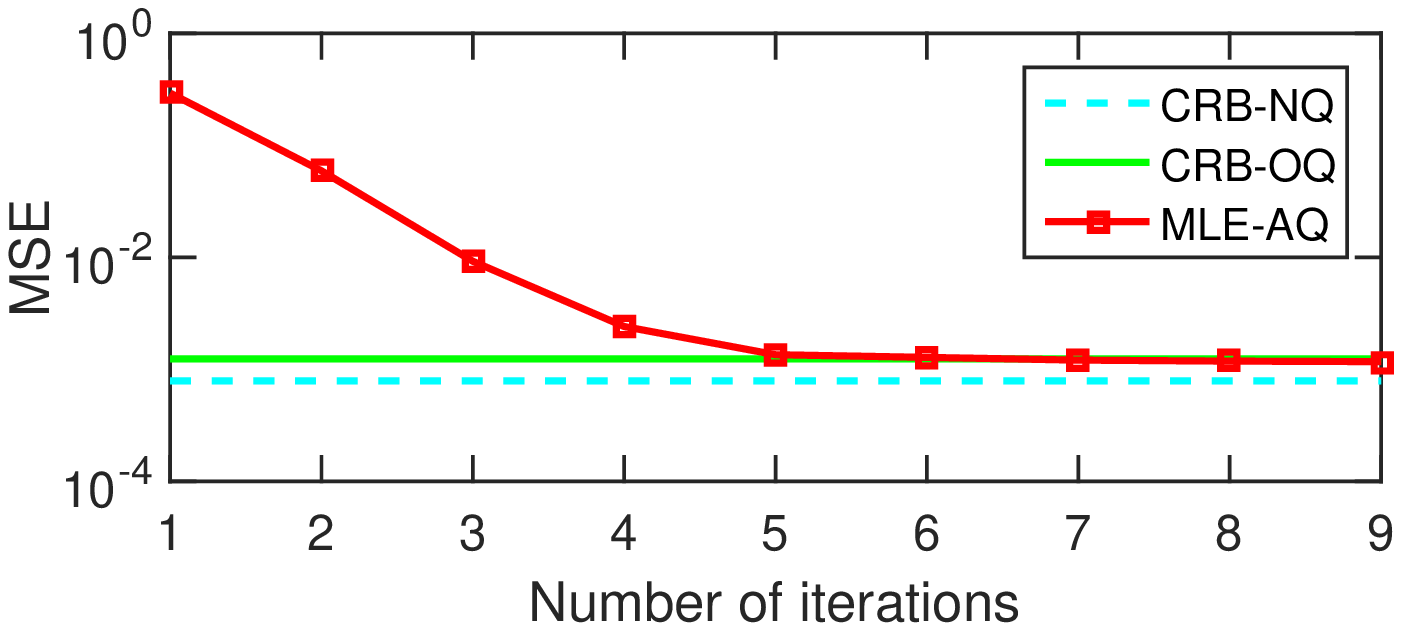}\\
(b). $K=16$, $L=40$ and $\text{SNR}=15$ dB.
\end{tabular}
  \caption{MSEs of the AQ scheme as a function of the number of iterations.}
   \label{fig3}
\end{figure}

We first examine the estimation performance of our proposed AQ
scheme which adaptively adjusts the thresholds based on the
previous estimate of the channel. Fig. \ref{fig3} plots the
mean-squared errors (MSEs) vs. the number of iterations for the AQ
scheme, where we set $K=8$, $L=32$ for Fig. (a) and $K=16$, $L=40$
for Fig. (b). The SNR is set to 15dB. The MSE is calculated as
\begin{align}
\text{MSE}=\frac{1}{K M}
\|\boldsymbol{H}-\boldsymbol{\hat{H}}\|_F^2
\end{align}
To better illustrate the effectiveness of the AQ scheme, we also
include the CRB results in Fig. \ref{fig3}. in which the CRB-OQ,
given by (\ref{CRB-Q}), represents the theoretical lower bound on
the estimation errors of any unbiased estimator using optimal
thresholds for one-bit quantization, and the CRB-NQ, given by
(\ref{CRB-NQ}), represents the lower bound on the estimation
errors of any unbiased estimator which has access to the original
observations. From Fig. \ref{fig3}, we see that our proposed AQ
scheme approaches the theoretical lower bound CRB-OQ within only a
few (say, 5) iterations, and achieves performance close to the CRB
associated with the NQ scheme. This result demonstrates the
effectiveness of the AQ scheme in searching for the optimal
thresholds. In the rest of our simulations, we set the maximum
number of iterations, $i_{\text{max}}$, equal to 5 for the AQ
scheme.

%the theoretical lower bound on the estimation errors for one-bit
%massive MIMO systems

\begin{figure}[!t]
\centering
\includegraphics[width=3.5in]{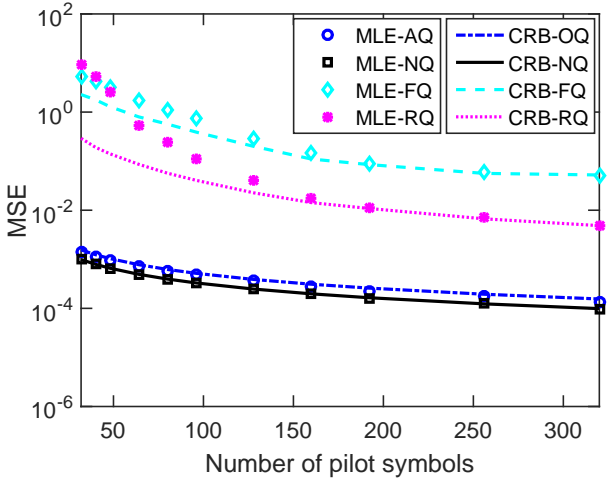}
\caption{MSEs vs. number of pilot symbols, where $K=8$ and
$\text{SNR}=15$ dB.} \label{fig4}
\end{figure}

\begin{figure}[!t]
\centering
\includegraphics[width=3.5in]{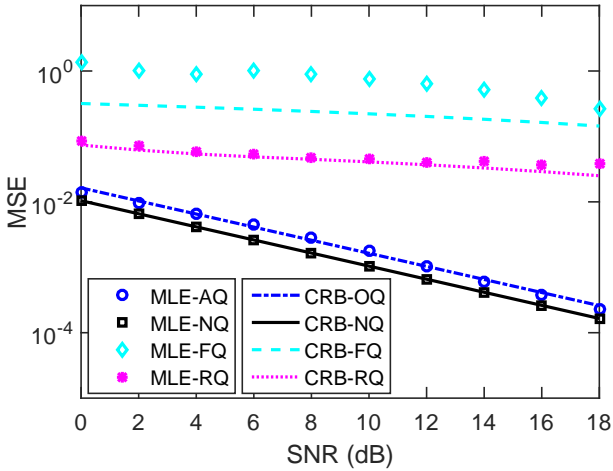}
\caption{MSEs vs. SNR(dB), where $K=8$ and $L=96$.} \label{fig5}
\end{figure}

%it takes about 300 pilot symbols for the RQ and an extremely long
%training sequence (up to thousands of pilot symbols) for the FQ
%scheme

We now compare the estimation performance of different schemes.
Fig. \ref{fig4} plots the MSEs of respective schemes as a function
of the number of pilot symbols, $L$, where we set $K=8$ and
$\text{SNR}=15\text{dB}$. The corresponding CRBs of these schemes
are also included. Note that the CRBs for the FQ and the RQ
schemes can be obtained by substituting the thresholds into
(\ref{CRB}). Results are averaged over $10^3$ independent runs,
with the channel and the pilot sequences randomly generated for
each run. From Fig. \ref{fig4}, we can see that the proposed AQ
scheme outperforms the FQ and RQ schemes by a big margin. This
result corroborates our analysis that an optimal choice of the
quantization thresholds helps achieve a substantial performance
improvement. In particular, the AQ scheme needs less than 30 pilot
symbols to achieve a decent estimation accuracy with a MSE of
0.01, while the FQ and RQ schemes require a much larger number of
pilot symbols to attain a same estimation accuracy. On the other
hand, we should note that although the AQ scheme has the potential
to achieve performance close to the NQ scheme, the implementation
of the AQ is more complicated since it involves an iterative
process to learn the optimal thresholds. In contrast, our proposed
RQ scheme is as simple as the FQ scheme to implement, meanwhile it
presents a clear performance advantage over the FQ scheme. We can
see from Fig. \ref{fig4} that the RQ requires about 100 symbols to
achieve a MSE of 0.1, whereas the FQ needs about 250 pilot symbols
to reach a same estimation accuracy. The reason why the RQ
performs better than the FQ is that some of the thresholds
produced according to (\ref{multi-thresholding}) are likely to be
close to the optimal thresholds. In Fig. \ref{fig5}, we plot the
MSEs of respective schemes under different SNRs, where we set
$K=8$ and $L=96$. Similar conclusions can be made from Fig.
\ref{fig5}.

\begin{figure}[!t]
\centering
\includegraphics[width=3.5in]{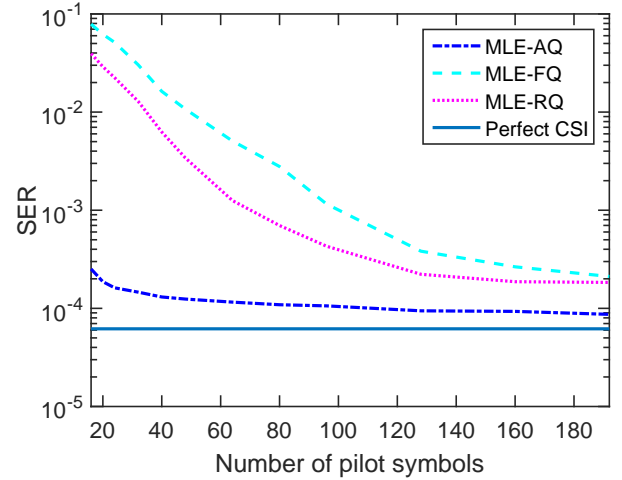}
\caption{SERs vs. number of pilot symbols, where $K=8$, $M=64$ and
$\text{SNR}=5\text{dB}$.} \label{fig6}
\end{figure}

\begin{figure}[!t]
\centering
\includegraphics[width=3.5in]{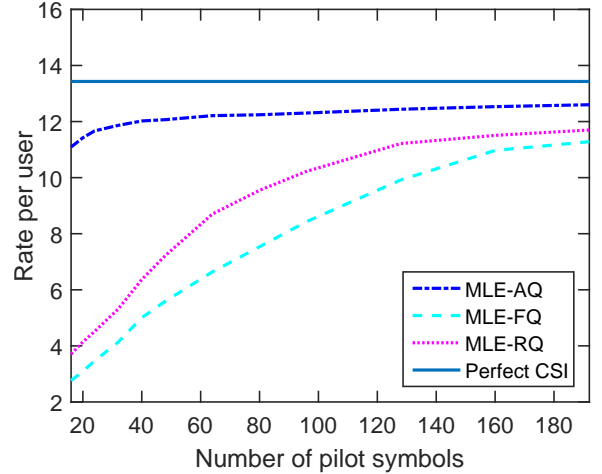}
\caption{Achievable rates vs. number of pilot symbols, where
$K=8$, $M=64$ and $\text{SNR}=5\text{dB}$.} \label{fig7}
\end{figure}

%to provide a benchmark for evaluating the performance of different
%schemes

Next, we examine the effect of channel estimation accuracy on the
symbol error rate (SER) performance. For each scheme, after the
channel is estimated, a near maximum likelihood detector
\cite{ChoiMo16} developed for one-bit massive MIMO is adopted for
symbol detection. For a fair comparison, in the symbol detection
stage, the quantization thresholds are all set equal to zero, as
assumed in \cite{ChoiMo16}. In our experiments, QPSK symbols are
transmitted by all users. Fig. \ref{fig6} plots the SERs of
respective schemes vs. the number of pilot symbols, where we set
$K=8$, $M=64$, and $\text{SNR}=5\text{dB}$. Results are averaged
over all $K$ users. The SER performance obtained by assuming
perfect channel knowledge is also included. It can be seen that
the SER performance improves as the number of pilot symbols
increases, which is expected since a more accurate channel
estimate can be obtained when more pilot symbols are available for
channel estimation. We also observe that the AQ scheme, using a
moderate number (about 120 symbols that is only 15 times the
number of users) of pilot symbols, can achieve SER performance
close to that attained by assuming perfect channel knowledge.
Moreover, the SER results, again, demonstrate the superiority of
the RQ over the FQ scheme. In order to attain a same SER, say,
$10^{-3}$, the RQ requires about 60 pilot symbols, whereas the FQ
requires about 100 pilot symbols.

In Fig. \ref{fig7}, the achievable rates of respective schemes vs.
the number of pilot symbols are depicted, where we set $K=8$,
$M=64$, and $\text{SNR}=5\text{dB}$. The achievable rate for the
$k$th user is calculated as \cite{MollenChoi17}
\begin{align}
R_k \triangleq \log_2 \left( 1+ \frac{
|\mathbb{E}\left[s_k^*(t)\hat{s}_k(t)\right]|^2 }
{\mathbb{E}\left[|\hat{s}_k(t)|^2\right] -
|\mathbb{E}\left[s_k^*(t)\hat{s}_k(t)\right]|^2} \right)
\end{align}
where $s_k(t)$ is the transmit symbol of the $k$th user at time
$t$, $()^{*}$ denotes the conjugate, and $\hat{s}_k(t)$ is the
estimated symbol of $s_k(t)$, which is obtained via the near
maximum likelihood detector by using the channel estimated by
respective schemes. The achievable rate we plotted is averaged
over all $K$ users. It can be seen that, even with a moderate
number of pilot symbols (about 5 times the number of users), the
AQ scheme can provide an achievable rate close to that of the
perfect CSI case, whereas the achievable rates attained by the
other two schemes are far below the level of the AQ scheme.
Compared to the FQ, the RQ scheme achieves an increase of about 30
percent in the achievable rate.

%Similar conclusions can be reached from this figure.

\section{Conclusions} \label{sec:conclusion}
Assuming one-bit ADCs at the BS, we studied the problem of one-bit
quantization design and channel estimation for uplink multiuser
massive MIMO systems. Specifically, based on the derived CRB
matrix, we examined the impact of quantization thresholds on the
channel estimation performance. Our theoretical analysis revealed
that using one-bit ADCs can achieve an estimation error close to
that attained by using infinite-precision ADCs, given that the
quantization thresholds are optimally set. Our analysis also
suggested that the optimal quantization thresholds are dependent
on the unknown channel parameters. We developed two practical
quantization design schemes, namely, an adaptive quantization
scheme which adaptively adjusts the thresholds such that the
thresholds converge to the optimal thresholds, and a random
quantization scheme which randomly generates a set of
non-identical thresholds based on some statistical prior knowledge
of the channel. Simulation results showed that the proposed
quantization schemes achieved a significant performance
improvement over the fixed quantization scheme that uses a fixed
(typically zero) quantization threshold, and thus can help
substantially reduce the training overhead in order to attain a
same estimation accuracy target.

\useRomanappendicesfalse
\appendices

\section{Proof of Concavity of The Log-Likelihood Function (\ref{log-PMF})} \label{appA}
It can be easily verified that
$f_{w}(\boldsymbol{a}_n^{T}\boldsymbol{h}-\tau_n)$ is log-concave
in $\boldsymbol{h}$ since the Hessian matrix of $\log
f_{w}(\boldsymbol{a}_n^{T}\boldsymbol{h}-\tau_n)$, which is given
by
\begin{align}
\frac {\partial^2 \log
f_{w}(\boldsymbol{a}_n^{T}\boldsymbol{h}-\tau_n)} {{\partial
\boldsymbol{h} \partial \boldsymbol{h}^T}} = -
\frac{\boldsymbol{a}_n \boldsymbol{a}_n^{T}} {\sigma^2}
\end{align}
is negative semidefinite. Consequently the corresponding
cumulative density function (CDF) and complementary CDF (CCDF),
which are integrals of the log-concave function
$f_{w}(\boldsymbol{a}_n^{T}\boldsymbol{h}-\tau_n)$ over convex
sets $(-\infty,\tau_n)$ and $(\tau_n,\infty)$ respectively, are
also log-concave, and their logarithms are concave. Since
summation preserves concavity, $L(\boldsymbol{h})$ is a concave
function of $\boldsymbol{h}$.

\section{Proof of Theorem \ref{theorem1}} \label{appB}
Define a new variable $z_n\triangleq
\boldsymbol{a}_n^{T}\boldsymbol{h}$ and define
\begin{align}
l(z_n) & \triangleq  \frac{1-b_n}{2}  \mathrm{log} [1-F_{w}(z_n-\tau_n)]   \nonumber \\
& \quad + \frac{1+b_n}{2} \mathrm{log} [F_{w}(z_n-\tau_n)].
\end{align}
The first and second-order derivative of $L(\boldsymbol{h})$ are
given by
\begin{align}
\frac {\partial L(\boldsymbol{h})} {\partial \boldsymbol{h}} =
\sum_{n=1}^{N} \frac {\partial l(z_n)} {\partial z_n} \frac
{\partial z_n} {\partial \boldsymbol{h}} = \sum_{n=1}^{N} \frac
{\partial l(z_n)} {\partial z_n} \boldsymbol{a}_n
\end{align}
and
\begin{align}
\frac {\partial^2 L(\boldsymbol{h})} {\partial \boldsymbol{h}
\partial \boldsymbol{h}^T}
&= \sum_{n=1}^{N} \boldsymbol{a}_n \frac {\partial^2 l(z_n)}
{\partial z_n^2}
\frac {\partial z_n} {\partial \boldsymbol{h}^T} \nonumber \\
&= \sum_{n=1}^{N} \frac {\partial^2 l(z_n)} {\partial z_n^2}
\boldsymbol{a}_n \boldsymbol{a}_n^T .
\end{align}
where
\begin{align}
\frac {\partial l(z_n)} {\partial z_n} &= \frac{1-b_n}{2}  \frac{f_{w}(z_n-\tau_n)}{F_{w}(z_n-\tau_n)-1}   \nonumber \\
& \quad + \frac{1+b_n}{2}
\frac{f_{w}(z_n-\tau_n)}{F_{w}(z_n-\tau_n)} \label{deriv1}
\end{align}
and
\begin{align}
\frac {\partial^2 l(z_n)} {\partial z_n^2} =& \frac{1-b_n}{2}
\bigg[
\frac{f'_{w}(z_n-\tau_n)}{F_{w}(z_n-\tau_n)-1} \nonumber \\
& -\frac{f_{w}^2 (z_n-\tau_n)}{(F_{w}(z_n-\tau_n)-1)^2} \bigg] + \frac{1+b_n}{2}   \nonumber \\
&  \cdot \bigg[ \frac{f'_{w}(z_n-\tau_n)}{F_{w}(z_n-\tau_n)} -
\frac{f_{w}^2 (z_n-\tau_n)}{F_{w}^2 (z_n-\tau_n)} \bigg]
\label{deriv2}
\end{align}
where $f_{w}(x)$ denotes the probability density function (PDF) of
$w_n$, and $f'_{w}(x)\triangleq\frac{\partial f_{w}(x)}{\partial
x}$.

Therefore, the Fisher information matrix (FIM) of the estimation
problem is given as
\begin{align}
J(\boldsymbol{h}) & = -E \left[\frac {\partial^2
L(\boldsymbol{h})} {\partial \boldsymbol{h} \partial
\boldsymbol{h}^T} \right]
= - \sum_{n=1}^{N} E_{b_n} \left[ \frac {\partial^2 l(z_n)} {\partial z_n^2} \right]
\boldsymbol{a}_n \boldsymbol{a}_n^T  \nonumber \\
& \stackrel {(a)}{=} \sum_{n=1}^{N} \frac {f_{w}^2
(\boldsymbol{a}_n^{T}\boldsymbol{h}-\tau_n)}
{F_{w}(\boldsymbol{a}_n^{T}\boldsymbol{h}-\tau_n)(1-F_{w}(\boldsymbol{a}_n^{T}\boldsymbol{h}-\tau_n))}
\boldsymbol{a}_n \boldsymbol{a}_n^T
\end{align}
where $E_{b_n}[\cdot]$ denotes the expectation with respect to the
distribution of $b_n$, and $(a)$ follows from the fact that
$b_{n}$ is a binary random variable with
$P(b_{n}=1|\tau_n,z_n)=F_{w}(z_n-\tau_n)$ and
$P(b_{n}=-1|\tau_n,z_n)=1-F_{w}(z_n-\tau_n)$. This completes the
proof.

%\begin{align}
%E\left[b_n\right]=2F_{w}(\boldsymbol{a}_n^{T}\boldsymbol{h}-\tau_n)-1
%\end{align}
%which can be derived from the bivariate distribution of $b_m$
%depicted in (\ref{probbn1}) and (\ref{probbn2}).

%Accordingly, the CRB matrix for the estimation problem
%(\ref{quantized-data-model-vector}) is the inverse of the FIM as
%\begin{align}
%\mathrm{CRB}(\boldsymbol{h}) = J^{-1}(\boldsymbol{h}) = \left(
%\sum_{n=1}^{M} g(\tau_n,\boldsymbol{a}_n) \boldsymbol{a}_n
%\boldsymbol{a}_n^T \right)^{-1} .
%\end{align}

\section{Proof of Proposition \ref{proposition1}} \label{appE}
Before proceeding, we first introduce the following lemma.
\newtheorem{lemma}{Lemma}
\begin{lemma} \label{lemma2}
For $x\ge 0$, define
\begin{align}
\bar{F}(x) \triangleq \int_{0}^{x} f(u)\mathrm{d}u
\end{align}
where $f(\cdot)$ denotes the PDF of a real-valued Gaussian random
variable with zero-mean and unit variance. We have $\bar{F}(x)$
upper bounded by
\begin{align}
\bar{F}(x) \le \frac{1}{2} \sqrt{ 1-e^{-\frac{2x^2}{\pi}} } .
\end{align}
\end{lemma}
\begin{proof}
See Appendix \ref{appF}.
\end{proof}

Define the function
\begin{align}
\bar{g}(x) \triangleq \frac {f^2 (x)} {F(x)(1-F(x))}
\end{align}
where $f(\cdot)$ and $F(\cdot)$ denotes the PDF and CDF of a
real-valued Gaussian random variable with zero-mean and unit
variance respectively. Invoking Lemma \ref{lemma2}, we have
\begin{align}
\bar{g}(x) =\frac {f^2 (x)} {\frac{1}{4}-\bar{F}^2(x)} \le
\frac{2}{\pi} e^{-(1-\frac{2}{\pi})x^2} \le \frac{2}{\pi}.
\end{align}
and $\bar{g}(x) =\frac{2}{\pi}$ if and only if $x=0$. Noting that
\begin{align}
\frac{1}{\sigma^2} \bar{g} \left(
\frac{\boldsymbol{a}_n^{T}\boldsymbol{h}-\tau_n}{\sigma} \right) =
g(\tau_n,\boldsymbol{a}_n) .
\end{align}
Therefore $g(\tau_n,\boldsymbol{a}_n)$ attains its maximum when
$\tau_n=\boldsymbol{a}_n^{T}\boldsymbol{h}$. The proof is
completed here.

\section{Proof of Lemma \ref{lemma2}} \label{appF}
Define two i.i.d. Gaussian random variables with zero-mean and
unit variance, namely, $X$ and $Y$. The joint distribution
function of $X$ and $Y$ is $f_{XY}(x,y)=f(x)f(y)$. Define two
regions $D_1 \triangleq \{(u,v)\mid 0 \le u \le x , 0 \le v \le x
\}$ and $D_2 \triangleq \{(u,v)\mid u \ge 0 , v \ge 0, u^2+v^2\le
\frac{4x^2}{\pi} \}$. Obviously, the areas of $D_1$ and $D_2$ are
the same, i.e., $\mu(D_1)=\mu(D_2)$, where $\mu(\cdot)$ denote the
area of a region. The probabilities of $(X,Y)$ belonging in these
two regions can be computed as
\begin{align}
P((X,Y)\in D_1) &= \iint_{D_1} f_{XY}(u,v) \mathrm{d}u\mathrm{d}v \nonumber \\
&= \bar{F}^2(x) \\
P((X,Y)\in D_2) &= \iint_{D_2} f_{XY}(u,v) \mathrm{d}u\mathrm{d}v \nonumber \\
&= \frac{1}{4} \left(1-e^{-\frac{2x^2}{\pi}}\right)
\end{align}
Let $S_1\setminus S_2$ denote the set obtained by excluding
$S_2\cap S_1$ from $S_1$. Clearly, we have
\begin{align}
\mu(D_1 \setminus D_2)=\mu(D_2 \setminus D_1) \label{appF:eqn1}
\end{align}
Also, according to the definition of $D_1$ and $D_2$, we have
\begin{align}
f_{XY}(u,v)&\le \frac{1}{2\pi} e^{-\frac{2x^2}{\pi}}, \quad (u,v)\in D_1 \setminus D_2 \\
f_{XY}(u,v)&\ge \frac{1}{2\pi} e^{-\frac{2x^2}{\pi}}, \quad
(u,v)\in D_2 \setminus D_1 \label{appF:eqn2}
\end{align}
Combining (\ref{appF:eqn1})--(\ref{appF:eqn2}), we arrive at
\begin{align}
\iint_{D_1 \setminus D_2} f_{XY}(u,v) \mathrm{d}u\mathrm{d}v \le
\iint_{D_2 \setminus D_1} f_{XY}(u,v) \mathrm{d}u\mathrm{d}v
\end{align}
From the above inequality, we have $P((X,Y)\in D_1)\le P((X,Y)\in
D_2)$, i.e.
\begin{align}
\bar{F}^2(x)\leq\frac{1}{4}
\left(1-e^{-\frac{2x^2}{\pi}}\right)\Rightarrow \bar{F}(x) \le
\frac{1}{2} \sqrt{ 1-e^{-\frac{2x^2}{\pi}} }
\end{align}
This completes the proof.

\section{Proof of Theorem \ref{theorem2}} \label{appC}
Note that from the constraint
$\text{tr}(\boldsymbol{X}\boldsymbol{X}^H)\leq P$, we can easily
derive that
\begin{align}
\text{tr}(\boldsymbol{A}^T \boldsymbol{A})\leq 2M P
\end{align}
To prove Theorem \ref{theorem2}, let us first consider a new
optimization that has the same objective function as (\ref{opt2})
while with a relaxed constraint:
\begin{align}
\min_{\boldsymbol{A}} \quad & \frac{\pi\sigma^2}{2}\text{tr}
\left\{ \left( \boldsymbol{A}^T \boldsymbol{A} \right)^{-1}
\right\}
\nonumber\\
\text{s.t.} \quad & \text{tr}(\boldsymbol{A}^T\boldsymbol{A})\leq
2M P  \label{appC:opt1}
\end{align}
Clearly, the feasible region defined by the constraints in
(\ref{opt2}) is a subset of that defined by (\ref{appC:opt1}).
Since $\text{tr}(\boldsymbol{Z}^{-1})$ is convex over the set of
positive definite matrix, the optimization (\ref{appC:opt1}) is
convex. Its optimum solution is given as follows.
\begin{lemma} \label{lemma1}
Consider the following optimization problem
\begin{align}
\min_{\boldsymbol{Z}}\quad &\text{tr}(\boldsymbol{Z}^{-1})
\nonumber\\
\text{s.t.}\quad & \text{tr}(\boldsymbol{Z})\leq P_0
\label{appC:opt2}
\end{align}
where $\boldsymbol{Z}\in\mathbb{R}^{p\times p}$ is positive
definite. The optimum solution to (\ref{appC:opt2}) is given by
$\boldsymbol{Z}=(P_0/p)\boldsymbol{I}$ and the minimum objective
function value is $p^2/P_0$.
\end{lemma}
\begin{proof}
See Appendix \ref{appD}.
\end{proof}

From Lemma \ref{lemma1}, we know that any $\boldsymbol{A}$
satisfying
\begin{align}
\boldsymbol{A}^T \boldsymbol{A} = (P/K) \boldsymbol{I}
\label{appC:eqn1}
\end{align}
is an optimal solution to (\ref{appC:opt1}). Note that the set of
feasible solutions (\ref{appC:opt1}) subsumes the feasible
solution set of (\ref{opt2}). Hence, if the optimal solution to
(\ref{appC:opt1}) is meanwhile a feasible solution of
(\ref{opt2}), then this solution is also an optimal solution to
(\ref{opt2}). It is easy to verify that if (\ref{theorem2:eqn1})
holds valid, the resulting $\boldsymbol{A}$ satisfies
(\ref{appC:eqn1}) and is thus an optimal solution to
(\ref{appC:opt1}). As a consequence, it is also an optimal
solution to (\ref{opt2}). This completes the proof.

%Hence the global minimum of (\ref{appC:opt1}) places a lower bound
%on the minimum achievable objective function value of
%(\ref{opt2}). Also,

\section{Proof of Lemma \ref{lemma1}} \label{appD}
Let $\boldsymbol{Z}=\boldsymbol{U}\boldsymbol{D}\boldsymbol{U}^T$
denote the eigenvalue decomposition of $\boldsymbol{Z}$, where
$\boldsymbol{U}\in\mathbb{R}^{p\times p}$ and
$\boldsymbol{D}\in\mathbb{R}^{p\times p}$. By replacing
$\boldsymbol{Z}$ with
$\boldsymbol{U}\boldsymbol{D}\boldsymbol{U}^T$, the optimization
(\ref{appC:opt2}) is reduced to determining the diagonal matrix
$\boldsymbol{D}\triangleq \text{diag}(d_1,\dots,d_{p})$
\begin{align}
\min_{\{d_i\}} \ & \sum_{i=1}^{p} \frac{1}{d_i}
\nonumber\\
\text{s.t.} \ \,
& \sum_{i=1}^{p} {d_i} \leq P_0 \nonumber\\
& d_i > 0, \qquad \forall i\in\{1,\dots,p\}\label{opt5}
\end{align}
The Lagrangian function associated with (\ref{opt5}) is given by
\begin{align}
L(d_i;\lambda;\nu_i)=\sum_{i=1}^{p} \frac{1}{d_i} + \lambda \left(
\sum_{i=1}^{p} {d_i} - P_0 \right) - \sum_{i=1}^{p} {\nu_i d_i}
\end{align}
with KKT conditions \cite{BoydVandenberghe03} given as
\begin{align}
-\frac{1}{d_i^2}+\lambda-\nu_i=0 & , \quad \forall i \nonumber\\
\lambda\left(\sum_{i=1}^{p} {d_i} - P_0\right)  =0 & \nonumber\\
\lambda\geq 0 & \nonumber\\
\nu_i d_i=0 & ,\quad \forall i \nonumber\\
d_i >0 & , \quad \forall i \nonumber\\
\nu_i\ge 0 & , \quad \forall i \nonumber
\end{align}
From the last three equations, we have $\nu_i=0$, $\forall i$.
Then from the first equation we have
\begin{align}
\lambda=\frac{1}{d_i^2}>0 \label{lambda}
\end{align}
and
\begin{align}
d_1=d_2=\dots=d_{p}.
\end{align}
From (\ref{lambda}) and the second equation, we have
$\sum_{i=1}^{p} {d_i} - P_0 =0$, from which $d_i$ can be readily
solved as $d_i=P_0/p$, $\forall i$, i.e., the optimal
$\boldsymbol{D}$ is given by $\boldsymbol{D}^{\star}= (P_0/p)
\boldsymbol{I}$. Consequently we have $\boldsymbol{Z}^{\star}=
(P_0/p) \boldsymbol{I}$. This completed the proof.

\bibliography{newbib}
\bibliographystyle{IEEEtran}

\end{document}